\documentclass[11pt,a4paper,reqno]{amsart}%
\usepackage{amsthm,amsmath,amsfonts,amssymb,xcolor,amsxtra,appendix,bookmark,dsfont,bm,mathrsfs}
\usepackage[latin1]{inputenc}
\usepackage{color} 
\usepackage{amssymb}\usepackage{graphicx}
\usepackage{tikz}
\usepackage{hyperref}
\theoremstyle{plain}
\newtheorem{theorem}{Theorem}
\newtheorem{lemma}[theorem]{Lemma}

\newtheorem{proposition}[theorem]{Proposition}

\newcommand{\RR}{\mathbb{R}}

\newcommand{\beqn}{\begin{equation}}
	\newcommand{\eeqn}{\end{equation}}
\newcommand{\bear}{\begin{eqnarray}}
	\newcommand{\eear}{\end{eqnarray}}
\newcommand{\bean}{\begin{eqnarray*}}
	\newcommand{\eean}{\end{eqnarray*}}

\newcommand{\cE}{\omega}
\newcommand{\p}{{\bf{p}}}

\theoremstyle{remark}

\allowdisplaybreaks

\title[wave kinetic equation in oceanography]{On a wave kinetic equation with resonance broadening in  oceanography and atmospheric sciences}

\author[Y. H. Kim]{Young Ho Kim
}
\address{Department of Mathematics, Texas A\&M University, College Station, TX 77843, USA. }
\email{yhkim@tamu.edu} 

\author[Y. V. Lvov]{Yuri V. Lvov
}
\address{Department of Mathematical Sciences, 
	Rensselaer Polytechnic Institute, Troy, NY 12180,  USA. }
\email{ lvovy@rpi.edu} 

\author[L. M. Smith]{Leslie M. Smith
}
\address{Department of Mathematics and Department of Engineering Physics, University of Wisconsin-Madison, Madison, WI 53706, USA. }
\email{ lsmith@math.wisc.edu} 

\author[M.-B. Tran]{Minh-Binh Tran}
\address{Department of Mathematics, Texas A\&M University, College Station, TX 77843, USA}
\email{minhbinh@tamu.edu} 
\thanks{Y.H. K. and M.-B. T are  funded in part by  a   Humboldt Fellowship,   NSF CAREER  DMS-2303146, and NSF Grants DMS-2204795, DMS-2305523,  DMS-2306379. }

\begin{document}

\maketitle
\begin{abstract} 

 In this work, we study a three-wave kinetic equation with resonance broadening arising from the theory of stratified ocean flows. Unlike \cite{GambaSmithBinh}, we employ a different formulation of the resonance broadening, which makes the present model more suitable for ocean applications. We establish the global existence and uniqueness of strong solutions to the new resonance broadening kinetic equation.
 \end{abstract}

{\bf Keywords:} wave (weak) turbulence theory, wave-wave interactions, stratified fluids, oceanography, near-resonance\\

{\bf MSC:} {35B05, 35B60, 82C40}

 \tableofcontents

\section{Introduction}

During the last few decades, wave-wave interactions in continuously stratified fluids have been an important subject of intensive research in oceanography and atmospheric sciences. One of the most important discoveries  in understanding  such wave-wave interactions is the observation of a nearly universal internal-wave energy spectrum in the ocean, first described by Garrett and Munk (cf. \cite{garrett1975space,garrett1979internal,cairns1976internal}). The existence of such a universal spectrum is the result of nonlinear interactions of waves with different wavenumbers,  interacting in triads (cf. \cite{waleffe1992nature}). Moreover, 
resonant triads are expected to dominate the dynamics for weak nonlinearity  (cf. \cite{mccomasbretherton77}).

Resonant wave interactions can be described by Zakharov kinetic equations (cf. \cite{zakharov2012kolmogorov,Nazarenko:2011:WT,majda1997one,cai1999spectral,zakharov1967weak,zakharov1968stability}), which reads
\begin{equation}\label{WeakTurbulenceInitiala}
\begin{aligned}
\partial_tf(t,\p) + \mu_\p f(t,\p) \ =& \ \mathbb{C}^{exact}[f](t,\p), \ \ \ f(0,\p)=f_0(\p),
\end{aligned}
\end{equation}
where  $f(t,\p)$ is the nonnegative wave density  at  wavenumber $\p \in \RR^d$, $d \ge 2$. Following \cite{zakharov1967weak}, 
$\mu_\p f=2\nu|\p|^\gamma f$ ($\gamma > 2$) 
is  the viscous damping term, and $\nu$ is the viscosity coefficient. The equation is a  three-wave kinetic one, in which the collision operator is of the form 
\begin{equation}\label{def-Qfa}\mathbb{C}^{exact}[f](\p) \ = \ \iint_{\mathbb{R}^{2d}} \Big[ \mathcal{N}^{exact}_{\p,\p_1,\p_2}[f] - \mathcal{N}^{exact}_{\p_1,\p,\p_2}[f] - \mathcal{N}^{exact}_{\p_2,\p,\p_1}[f] \Big] d\p_1d\p_2 \end{equation}
with $$\begin{aligned}
\mathcal{N}^{exact}_{\p,\p_1,\p_2} [f]:=   |\bar V_{\p,\p_1,\p_2}|^2\delta(\p-\p_1-\p_2)\delta(\omega_\p -\omega_{\p_1}-\omega_{\p_2})(f_1f_2-ff_1-ff_2),
\end{aligned}
$$
and  we use the short-hand notation $f = f(t,\p)$ and $f_j = f(t,\p_j)$. The  collision kernel $V_{\p,\p_1,\p_2}$ is of the form (cf. \cite{lvov2004noisy,connaughton2001discreteness,lvov2004hamiltonian,lvov2012resonant,l1997statistical})
\begin{equation}\label{def-VVa}
\begin{aligned}
\bar V_{\p,\p_1,\p_2}  \ = \mathfrak{C}  
\left({|\p||\p_1||\p_2|}\right)^\frac12
\text{,}
\end{aligned}
\end{equation}
where $\mathfrak{C}$ is some physical constant, which is set to be $1$. 

 The equations   describe the 
spectral energy transfer on the resonant manifold, which is a set of wave vectors $\p$, $\p_1$, $\p_2$ satisfying
\begin{equation}
\label{ResonantManifold}
\p=\p_1+\p_2,\ \ \ \ \ \ \omega_\p=\omega_{\p_1}+\omega_{\p_2},
\end{equation}
where the frequency $\omega$ is given by 
the dispersion relation between the wave frequency $\omega$ and the  wavenumber $\p$ \begin{equation}\label{Def:DispersionLaw}
\omega_\p=\sqrt{F^2+\frac{g^2}{\rho_0^2N^2}\frac{|\p|^2}{m^2}},
\end{equation}
where $F$  is the Coriolis parameter, $N$ is the buoyancy frequency, $m$ is the reference vertical wave number determined from observations, $g$ is the gravitational constant, $\rho_0$ is the constant reference value for the density. Let us set $\Lambda_1=F^2$ and  $\Lambda_2=g^2/(m^2\rho_0^2N^2)$, such that
\begin{equation}\label{Def:DispersionLaw2}
\omega_\p=\sqrt{\Lambda_1+\Lambda_2{|\p|}^2}.
\end{equation}

However, it is known that exact resonances defined by $\omega_\p=\omega_{\p_1}+\omega_{\p_2}$ do not capture some important physical effects, some authors have included more physics by  allowing near-resonant interactions (cf. 
\cite{connaughton2001discreteness,lee2007formation,lvov2012resonant,l1997statistical,lvov2004hamiltonian,lvov2004noisy,lvov2010oceanic,newell69,smith2005near,remmel2009new,remmel2010nonlinear}), 
defined as
\begin{equation}
\label{NearResonantManifold}
\p=\p_1+\p_2,\ \ \ \ |\omega_\p-\omega_{\p_1}-\omega_{\p_2}|<\theta(f,\p), 
\end{equation}
where $\theta$ accounts for broadening of the resonant surfaces and  is a function of  the wave density $f$ and the wave number $\p$ (cf. \cite{chekhlov96,huang2000,lee2007formation,remmel2014nonlinear,smith2001,smith2005near,smith1999transfer,smith2002generation}).

In the previous work \cite{GambaSmithBinh}, we considered the following near-resonance turbulence kinetic equation 
\cite{connaughton2001discreteness,l1997statistical,lvov2004hamiltonian,lvov2004noisy,lvov2012resonant}), 

\begin{equation}\label{def-Qfb}\mathbb{C}^{Broaden}[f](\p) \ = \ \iint_{\mathbb{R}^{2d}} \Big[ \mathcal{N}^{Broaden}_{\p,\p_1,\p_2}[f] - \mathcal{N}^{Broaden}_{\p_1,\p,\p_2}[f] - \mathcal{N}^{Broaden}_{\p_2,\p,\p_1}[f] \Big] d\p_1d\p_2 \end{equation}
with $$\begin{aligned}
\mathcal{N}^{Broaden}_{\p,\p_1,\p_2} [f]:=   |\bar V_{\p,\p_1,\p_2}|^2\delta(\p-\p_1-\p_2)\mathcal{L}^{Broaden}_f(\omega_\p -\omega_{\p_1}-\omega_{\p_2})(f_1f_2-ff_1-ff_2),
\end{aligned}
$$
and the operator $\mathcal{L}_f^{Broaden}$ is the Laurentian
\begin{equation}
\label{OperatorLb}
\mathcal{L}_f^{Broaden}(\Delta)=\frac{\bar\Gamma^f_{\p,\p_1,\p_2}}{\Delta^2+(\bar\Gamma^f_{\p,\p_1,\p_2})^2},
\end{equation}
with the condition that
$$\lim_{\Gamma^f_{\p,\p_1,\p_2}\to 0}\mathcal{L}^{Broaden}_f(\Delta)=\pi\delta(\Delta).$$
Moreover, the resonance broadening frequency $\Gamma^f_{\p,\p_1,\p_2}$ may be written
\begin{equation}
\label{Gamma}
\bar\Gamma^f_{\p,\p_1,\p_2}=\gamma_\p+\gamma_{\p_1}+\gamma_{\p_2},
\end{equation}
where $\gamma_\p$ is computed in \cite{l1997statistical} using a one-loop approximation:
\begin{equation}\label{GST}
\gamma_\p\backsim \mathfrak{c}|\p|^2\int_{\mathbb{R}_+}|\p|^2|f(t,|\p|)|d|\p|,\end{equation}
and $\mathfrak{c}$ is a physical constant, which can be normalized to be $1$.

However, the approximation \eqref{GST} is designed mainly for the acoustic dispersion relation 
$\omega(|\p|)=|\p|$, and thus is it serves mainly as a proof-of-concept.

A different approximation was proposed in \cite{polzin2017oceanic}, where $\gamma_\p$ is computed as \begin{equation}\label{Gammanew} \gamma_\p \backsim \mathfrak{c}_1 \max\{ \omega(|\p|) f(t,|\p|), \mathfrak{c}_2 \}, \end{equation} for some physical constants $\mathfrak{c}_1, \mathfrak{c}_2 > 0$.  
The approximation \eqref{Gammanew} is based on a class of three-wave interactions associated with induced diffusion in the ocean, where two wavenumbers are much larger in magnitude than the third wavenumber \cite{mccomasbretherton77}.
Using \eqref{Gammanew} in place of \eqref{GST} is expected to be a better approximation to describe some of the energy transfer influencing small-scale processes in the ocean interior, since \eqref{GST} is designed mainly for acoustic waves.

Following \cite{polzin2017oceanic}, we here use  
\eqref{Gammanew} in place of \eqref{GST} for $\gamma_\p$, and we consider the re-formulated kinetic equation

\begin{equation}\label{WeakTurbulenceInitial}
\begin{aligned}
\partial_tf(t,\p) + \mu_\p f(t,\p) \ =& \ \mathbb{C}[f](t,\p), \ \ \ f(0,\p)=f_0(\p),
\end{aligned}
\end{equation}
\begin{equation}\label{def-Qf}\mathbb{C}[f](\p) \ = \ \iint_{\mathbb{R}^{2d}} \Big[ \mathcal{N}_{\p,\p_1,\p_2}[f] - \mathcal{N}_{\p_1,\p,\p_2}[f] - \mathcal{N}_{\p_2,\p,\p_1}[f] \Big] d\p_1d\p_2 \end{equation}
with $$\begin{aligned}
\mathcal{N}_{\p,\p_1,\p_2} [f]:=   |V_{\p,\p_1,\p_2}|^2\delta(\p-\p_1-\p_2)\mathcal{L}_f(\omega_\p -\omega_{\p_1}-\omega_{\p_2})(f_1f_2-ff_1-ff_2),
\end{aligned}
$$
and the operator $\mathcal{L}_f$ is of the  form 
\begin{equation}
\label{OperatorL}
\mathcal{L}_f(\Delta)=  \frac{\Gamma^f_{\p,\p_1,\p_2}}{\Delta^2+(\Gamma^f_{\p,\p_1,\p_2})^2},  
\end{equation}
 Note that the formulation of $\Gamma^f_{k,k_1,k_2}$ is given
\begin{equation}
\Gamma^f_{\p,\p_1,\p_2}=\mathfrak{c}_1\max\{\omega(|\p|) f(t,|\p|),\mathfrak c_2\} + \mathfrak{c}_1\max\{\omega(|\p_1|)f(t,|\p_1|),\mathfrak c_2\} +\mathfrak{c}_1\max\{\omega(|\p_2|)f(t,|\p_2|),\mathfrak c_2\}.
\end{equation}
The kernel \eqref{def-VVa} is replaced by
\begin{equation}\label{def-VV}
	\begin{aligned}
		 V_{\p,\p_1,\p_2}  \ = \mathfrak{C}  
		\left( |\p|+|\p_1|+|\p_2| \right),
	\end{aligned}
\end{equation}
following \cite{polzin2017oceanic}.

{\it It is our goal to construct, for the first time, global unique  solutions in $L^1_m(\mathbb{R}^d)$ to \eqref{WeakTurbulenceInitial}.}

Let us mention that  the analysis of 3-wave kinetic equations  has been studied extensively across numerous physical contexts. Applications include Bose-Einstein condensates~\cite{cortes2020system, EPV, escobedo2023linearized1, escobedo2023linearized, escobedo2025local, ToanBinh,nguyen2017quantum,PomeauBinh,SofferBinh1,staffilani2025formation}, phonon interactions in crystal lattices~\cite{AlonsoGambaBinh, CraciunBinh, EscobedoBinh, GambaSmithBinh, tran2020reaction}, stratified ocean flows~\cite{GambaSmithBinh},  capillary waves~\cite{das2024numerical, nguyen2017quantum, soffer2020energy, walton2022deep, walton2023numerical, walton2024numerical}, and beam waves ~\cite{rumpf2021wave}.

We  split $\mathbb{C}$ as the sum of a gain and a loss operators:
\begin{equation}\label{GainLoss}
\mathbb{C}[f] \ = \ \mathbb{C}_{\mathrm{gain}}[f] \ - \ \mathbb{C}_\mathrm{loss}[f],
\end{equation}
as is done with the classical Boltzmann operator  for binary elastic interactions. Here, the gain operator is also defined by the positive contributions in the total rate of change in time of the collisional form $\mathbb{C}[f](t,\p)$ 
\begin{equation}\label{Qgain}
\begin{aligned}
\mathbb{C}_\mathrm{gain}[f] \ = & \ \iint_{\mathbb{R}^{d}\times \mathbb{R}^{d}}|V_{\p,\p_1,\p_2}|^2\delta(\p-\p_1-\p_2)\mathcal{L}_f(\omega_\p-\omega_{\p_1}-\omega_{\p_2})f_1f_2d\p_1d\p_2 \\
\ & +2\iint_{\mathbb{R}^{d}\times \mathbb{R}^{d}} |V_{\p_1,\p,\p_2}|^2\delta(\p_1-\p-\p_2)\mathcal{L}_f(\omega_{\p_1}-\omega_{\p}-\omega_{\p_2})(ff_1+f_1f_2)d\p_1d\p_2.
\end{aligned}
\end{equation}
and the loss operator models the negative contributions in the total
rate of change in time of the same collisional form $\mathbb{C}[f](t,\p)$
\begin{equation}\label{Qlosss}\mathbb{C}_\mathrm{loss}[f] \ = \ f\vartheta[f],\end{equation}

with $\vartheta[f]$ being the collision frequency or attenuation coefficient, defined by
\begin{equation}\label{Qloss}
\begin{aligned}
\vartheta[f](\p) \ = & \ 2\iint_{\mathbb{R}^{d}\times \mathbb{R}^{d}}|V_{\p,\p_1,\p_2}|^2\delta(\p-\p_1-\p_2)\mathcal{L}_f(\omega_\p-\omega_{\p_1}-\omega_{\p_2})f_1d\p_1d\p_2 \\
\ & +2\iint_{\mathbb{R}^{d}\times \mathbb{R}^{d}} |V_{\p_1,\p,\p_2}|^2\delta(\p_1-\p-\p_2)\mathcal{L}_f(\omega_{\p_1}-\omega_{\p}-\omega_{\p_2})f_2d\p_1d\p_2.
\end{aligned}
\end{equation}

For $m>0$, let $L^1_m(\RR^d)$ be the function space consisting of $g(\p)$ so that the norm 
$$ \| g\|_{L^1_m} : = \int_{\RR^d} |g(\p)| \cE_\p^m \; d\p $$
is finite.

For a given function $g$, we also define the $m$-th moment by
\begin{equation}\label{Def:MomentOrderk}
\mathcal{M}_m[g]=\int_{\mathbb{R}^d} g(\p)\cE^m_\p \,d\p.
\end{equation} 
Notice that when $g$ is positive $\mathcal{M}_n[g]$ and $\|g\|_{L^1_n}$ are equivalent.

We shall construct  global unique  solutions in $L^1_m(\mathbb{R}^d)$ to \eqref{WeakTurbulenceInitial}, or equivalently
\begin{equation}\label{WeakTurbulenceInitialReformed}
\partial_t f(t,\p) \ = \ \mathbb{C}_\mathrm{gain}[f](t,\p) \ - \ f(t,\p)\vartheta[f](t,\p) \ - 2\nu|\p|^\gamma f, \ \ \ f(0,\p)=f_0(\p).
\end{equation} 

Let us  define
$$\theta_* : = \widetilde{C}(\Lambda_1, \Lambda_2, \gamma, \nu)$$
where 
$\widetilde{C}>0$ is a constant depending on $\Lambda_1,\Lambda_2,\gamma,\nu$ 
to be defined later in Proposition \ref{Propo:MomentsPropa}. 
For any $\varsigma>1$ and $m,t >0$, 
we introduce $\Omega_t$ which includes functions 
$f \in L^{1}_{m+3}\big(\mathbb{R}^{d}\big) $ that satisfy 
\begin{equation}\label{Cons}
\begin{aligned}
&\mbox{(S1) Positivity of the set } \Omega_t: \ \ \ f\ge 0;\\
&\mbox{(S2) Upper bound of the set } \Omega_t: \ \ \   \| f\|_{L^1_{m+3}} \le c_0(t)
:=(2\varsigma+1)e^{\theta_* t}\text{.}
\end{aligned}
\end{equation}
Since $c_0(t)$ is an increasing function, 
$\Omega_t\subset \Omega_{t'}$ for $0\le t\le t'\le T$ and our main result is as follows. 

\begin{theorem}\label{Theorem:Main}
Let $N>0$, $\gamma > 2$, $T>0$, and let
\[
f_0(\p) \in \Omega_0 \cap B_*(O,\varsigma)
\]
for some $\varsigma > 1$, 
where $B_*(O,\varsigma)$ denotes the ball in $L^1_{m+3}(\mathbb{R}^d)$ centered at $O$ with radii $\varsigma$.

Then the weak turbulence equation \eqref{WeakTurbulenceInitial} admits a unique strong solution $f(t,\p)$ such that
\begin{equation}\label{the_theorem}
	0 \leq f(t,\p) \in C\!\left([0,T); L^1_m(\mathbb{R}^d)\right) 
	\cap C^1\!\left((0,T); L^1_m(\mathbb{R}^d)\right).
\end{equation}

Moreover, $f(t,\p)\in \Omega_T$ for all $t\in [0,T)$. 

Since $T$ can be chosen arbitrarily large, the weak turbulence equation \eqref{WeakTurbulenceInitial} has a unique global solution for all time $t>0$. 
\end{theorem}

	The proof of Theorem \ref{Theorem:Main} relies on the following abstract ODE theorem, inspired by previous works in quantum kinetic theory \cite{AlonsoGambaBinh,Bressan}.  
	
	Let $\mathfrak{E} = (\mathfrak{E}, \|\cdot\|)$ be a Banach space of real functions on $\mathbb{R}^d$, and let $(\mathfrak{F}, \|\cdot\|_*)$ be a Banach subspace of $\mathfrak{E}$ satisfying $\|u\| \le  C\|u\|_*\, \forall u \in \mathfrak{F}$ for some positive constant $C$. 
	Denote by $B(O,r)$ and $B_*(O,r)$ the balls centered at $O$ of radius $r>0$ with respect to the norms $\|\cdot\|$ and $\|\cdot\|_*$, respectively.  
	
	Suppose there exists a function $|\cdot|_* \colon \mathfrak{F} \to \mathbb{R}$ such that  
	\[
	|u|_* \le \|u\|_*, \quad \forall u \in \mathfrak{F}, \qquad 
	|u+v|_* \le |u|_* + |v|_*, \quad \forall u,v \in \mathfrak{F},
	\]
	and
	\[
	\Lambda |u|_* = |\Lambda u|_*, \quad \forall u \in \mathfrak{F},\ \Lambda \in \mathbb{R}_+.
	\]

\begin{theorem}\label{Theorem:ODE} 
	Let $[0,T]$ be a time interval, and let $\Omega_t$ $(t\in[0,T])$ be a family of bounded, closed subsets of $\mathfrak{F}$ such that $\Omega_t \subset \Omega_{t'}$ for $0 \le t \le t'$, and each $\Omega_t$ contains only nonnegative functions. Assume further that
	\[
	|u|_* = \|u\|_*, \quad \forall\, u \in \Omega_T.
	\]
	
	Moreover, for any sequence $\{u_n\}$ in $\Omega_T$, 
	\begin{equation}\label{LesbegueDominated}
		\text{if } u_n \geq 0,\ \|u_n\|_* \le C,\ \lim_{n\to\infty}\|u_n - u\| = 0,
		\quad\text{then}\quad 
		\lim_{n\to\infty}\|u_n - u\|_* = 0,
	\end{equation}
	for some constant $C>0$.
	
	Let $\varsigma>1$, and suppose $\mathcal{Q}:\Omega_T \to \mathfrak{E}$ is an operator satisfying the following properties. There exist constants $\eta,\theta_*,L>0$ such that:
	
	\begin{itemize}
		
		\item[$(\mathscr{A})$] \textbf{H\"{o}lder continuity.}  
		\begin{equation*}
			\|Q[u] - Q[v]\| \leq C\|u - v\|^{\beta}, \qquad \beta \in (0,1),\quad \forall\, u,v \in \Omega_T.
		\end{equation*}
		
		\item[$(\mathscr{B})$] \textbf{Sub-tangent condition.}  
		For each $u \in \Omega_T$, there exists $\xi_u>0$ such that for $0<\xi<\xi_u$, one can find $z \in B(u+\xi \mathcal{Q}[u],\delta)\cap \Omega_T \setminus \{u+\xi \mathcal{Q}[u]\}$ (for $\delta$ small enough) such that
		\begin{equation}\label{Coercivity}\begin{aligned}
				|z-u|_* &\le \tfrac{\theta_* \xi}{2}\|u\|_*
				\text{.}
		\end{aligned}\end{equation}
		
		\item[$(\mathscr{C})$] \textbf{One-sided Lipschitz condition.}  
		\begin{equation*}
			[Q[u] - Q[v], u - v] \leq L\|u - v\|, \qquad \forall\, u,v \in \Omega_T,
		\end{equation*}
		where
		\[
		[\varphi,\phi] := \lim_{h \to 0^-} h^{-1}\big(\|\phi + h\varphi\| - \|\phi\|\big).
		\]
		
	\end{itemize}
	
	In addition, 
	assume that $B\!\left(0,(2\varsigma+1)e^{\theta_*T}\right)\subset\Omega_T$.
	
	Then the equation
	\begin{equation}\label{Theorem_ODE_Eq}
		\partial_t u = Q[u] \quad \text{on } [0,T)\times \mathfrak{E}, 
		\qquad u(0) = u_0 \in \Omega_0 \cap B_*(O,\varsigma)
	\end{equation}
	admits a unique solution
	\[
	u \in C^1\big((0,T),\mathfrak{E}\big)\cap C\big([0,T),\Omega_T\big).
	\]
\end{theorem}

The proof of Theorem \ref{Theorem:Main} is given in Section \ref{Sec:Main}. 
The proof of Theorem \ref{Theorem:ODE} is given in Section \ref{Appendix}.

\section{A Preliminary estimate and estimates of $\mathbb{C}_{\mathrm{gain}}$} 

We start by proving the following preliminary estimate.

\begin{lemma}\label{Lemma:WeakFormulation}
	For any test function $\phi$ such that the integrals below are well defined, we have
	\[
	\int_{\mathbb{R}^d}\mathbb{C}[f](t,\p)\, \phi(\p)\, d\p
	= \iiint_{\mathbb{R}^{3d}} \mathcal{N}_{\p,\p_1,\p_2}[f]\,
	\big[\phi(\p) - \phi(\p_1) - \phi(\p_2)\big]\,
	d\p\, d\p_1\, d\p_2.
	\]
\end{lemma}

\begin{proof}
	By definition, the integral of the product of $\mathbb{C}[f]$ and $\phi$ can be written as
	\[
	\int_{\mathbb{R}^d}\mathbb{C}[f](t,\p)\, \phi(\p)\, d\p
	= \iiint_{\mathbb{R}^{3d}}
	\big[\mathcal{N}_{\p,\p_1,\p_2}
	- \mathcal{N}_{\p_1,\p,\p_2}
	- \mathcal{N}_{\p_2,\p,\p_1}\big]\,
	\phi(\p)\, d\p\, d\p_1\, d\p_2.
	\]
	Applying the change of variables $\p \leftrightarrow \p_1$ and $\p \leftrightarrow \p_2$ in the last two integrals on the right-hand side yields the desired result.
\end{proof}

Next, we prove the following estimate on the gain part of the collision operator $\mathbb{C}[g]$  defined in \eqref{GainLoss} and \eqref{Qgain}. 
\begin{lemma}\label{Propo:C12} Let $m \geq 0$. For any positive function $g \in L^1_{m+2}$, we have
	\begin{equation}\label{Propo:C12:1}
		\int_{\mathbb{R}^d}\mathbb{C}_{\mathrm{gain}}[g](\p)\,\cE^m_\p \, d\p
		\;\lesssim\; {\mathcal{M}_{m+2}[g]},
	\end{equation}
	where the implicit constant depends only on $\Lambda_1$ and $\Lambda_2$.
	
\end{lemma}

\begin{proof}
		By the same argument used to obtain the weak formulation in Lemma~\ref{Lemma:WeakFormulation}, 
		we have
		\begin{equation}\label{eq:weak-energy}
			\int_{\RR^d}\mathbb{C}[g](\p)\,\cE_\p^m \, d\p 
			= \iiint_{\RR^{3d}} \mathcal{N}_{\p,\p_1,\p_2}[g]\,
			\Big(\cE^m_{\p}-\cE^m_{\p_1}-\cE^m_{\p_2}\Big)\, d\p\, d\p_1\, d\p_2,
		\end{equation}
		where
		\[
		\mathcal{N}_{\p,\p_1,\p_2}[g] := |V_{\p,\p_1,\p_2}|^2 \,
		\delta(\p-\p_1-\p_2)\,\mathcal{L}(\omega_\p-\omega_{\p_1}-\omega_{\p_2})
		\,(g_1 g_2 + g g_1 + g g_2).
		\]

		\medskip
		\noindent\textit{Step 1. Splitting the gain term.}
		
		Since $\p_1$ and $\p_2$ are symmetric in the second integral we can write $(gg_1+gg_2)  \cE^m_{\p}$ as  
		$$
		\begin{aligned}
			& \ \int_{\RR^d}\mathbb{C}_{\mathrm{gain}}[g](\p) \cE_\p^m \; d\p=  \\
			= &\  C\iiint_{\mathbb{R}^{3d}} \delta(\p-\p_1-\p_2)
			(|\p|+|\p_1|+|\p_2|)^2
			\frac{\Gamma_{\p,\p_1,\p_2}^{g}}{
				(\omega_\p-\omega_{\p_1}-\omega_{\p_2})^2
				+(\Gamma_{\p,\p_1,\p_2}^g)^2
			}\\
			&\times g_1g_2\cE^m_{\p}d\p d\p_1d\p_2\\
			&\  + C\iiint_{\mathbb{R}^{3d}} \delta(\p-\p_1-\p_2)
			(|\p|+|\p_1|+|\p_2|)^2
			\frac{\Gamma_{\p,\p_1,\p_2}^{g}}{
				(\omega_\p-\omega_{\p_1}-\omega_{\p_2})^2
				+(\Gamma_{\p,\p_1,\p_2}^g)^2
			}\\
			&\times gg_1\left[\cE^m_{\p_1}+\cE^m_{\p_2}\right]d\p d\p_1d\p_2.
		\end{aligned}
		$$
		
		The fractional term in the above integral
		$$K:=
			(|\p|+|\p_1|+|\p_2|)^2
			\frac{\Gamma_{\p,\p_1,\p_2}^{g}}{
				(\omega_\p-\omega_{\p_1}-\omega_{\p_2})^2
				+(\Gamma_{\p,\p_1,\p_2}^g)^2
			}
		$$
		can be bounded as
		$$
		K\le 
		\frac{
			3(|\p|^2+|\p_1|^2+|\p_2|^2)
		}{\Gamma^g_{\p,\p_1,\p_2}}
		\le \frac{
			|\p|^2+|\p_1|^2+|\p_2|^2
		}{\mathfrak{c}_1\mathfrak{c}_2}
		\text{,}
		$$
		which yields the following bound on the integral 
		$$
		\begin{aligned}
			& \ \int_{\RR^d}\mathbb{C}_{\mathrm{gain}}[g](\p) \cE_\p^m \; d\p  \\
			\lesssim &\  \iiint_{\mathbb{R}^{3d}} \delta(\p-\p_1-\p_2)
			(|\p|^2+|\p_1|^2+|\p_2|^2)
			g_1g_2\cE^m_{\p} d\p d\p_1d\p_2\\
			& + \  \iiint_{\mathbb{R}^{3d}} \delta(\p-\p_1-\p_2)
			(|\p|^2+|\p_1|^2+|\p_2|^2)
			gg_1\Big [\cE^m_{\p_1}+\cE^m_{\p_2} \Big] d\p d\p_1d\p_2,
		\end{aligned}
		$$

		Let us rewrite the above inequality in the following equivalent form, where the right hand side is the sum of $\mathcal{A}_1$ and $\mathcal{A}_2$
		\begin{equation}\label{Propo:C12:E1}
			\begin{aligned}
				\int_{\RR^d}\mathbb{C}_{\mathrm{gain}}[g](\p) \cE_\p^m \; d\p  
				\lesssim &\  \mathcal{A}_1+\mathcal{A}_2,
			\end{aligned}
		\end{equation}
		where
		\begin{equation}\label{Propo:C12:E2}
			\begin{aligned}
				\mathcal{A}_1 := & \ \iiint_{\mathbb{R}^{3d}} \delta(\p-\p_1-\p_2)
				(|\p|^2+|\p_1|^2+|\p_2|^2)
				g_1g_2  \cE^m_{\p}d\p d\p_1d\p_2\\
				\mathcal{A}_2 := & \ \iiint_{\mathbb{R}^{3d}} \delta(\p-\p_1-\p_2)
				(|\p|^2+|\p_1|^2+|\p_2|^2)
				gg_1  \Big [\cE^m_{\p_1}+\cE^m_{\p_2} \Big] d\p d\p_1d\p_2.
			\end{aligned}
		\end{equation}

		\medskip
		\noindent\textit{Step 2. Estimate of $\mathcal{A}_1$.}
		
	Using the resonant condition  $\p = \p_1 + \p_2$, 
	$$\cE_\p =\sqrt{\Lambda_1+\Lambda_2 |\p|^2} \le \sqrt{\Lambda_1+\Lambda_2 (|\p_1|+|\p_2|)^2}$$
	$$< 2\sqrt{\Lambda_1+\Lambda_2 |\p_1|^2} + 2\sqrt{\Lambda_1+\Lambda_2 |\p_2|^2} = 2\cE_{\p_1} +2\cE_{\p_2},$$ 
	which, by the Cauchy-Schwarz inequality, yields
	$$\cE_\p^m  \lesssim (\cE_{\p_1}^m + \cE_{\p_2}^m),$$ 
	where the constant on the right hand side depends only on $\Lambda_1,\Lambda_2,N$.\\

	This inequality yields the following bound on $\mathcal{A}_1$
	$$
	\begin{aligned}
		\mathcal{A}_1 \lesssim & \ \iiint_{\mathbb{R}^{3d}} \delta(\p-\p_1-\p_2)
		(|\p|^2+|\p_1|^2+|\p_2|^2)
		g_1g_2  \Big [\cE^m_{\p_1}+\cE^m_{\p_2} \Big] d\p d\p_1d\p_2.
	\end{aligned}
	$$
	Integrating by $\p$ and using the definition of the Dirac function 
	$\delta(\p-\p_1-\p_2)$  yields
	\begin{align*} 
		\mathcal{A}_1 \lesssim & \ \iint_{\mathbb{R}^{2d}} 
		(|\p_1+\p_2|^2+|\p_1|^2+|\p_2|^2)
		g_1g_2  \Big [\cE^m_{\p_1}+\cE^m_{\p_2} \Big] d\p d\p_1d\p_2.
	\end{align*}

	Notice that 
	\begin{align*} 
		|\p|\leq\frac{\omega_{\p}}{\sqrt{\Lambda_2}}
		\text{,}\quad
		|\p_1|\leq\frac{\omega_{\p_1}}{\sqrt{\Lambda_2}}
		\text{,}\quad
		|\p_2|\leq\frac{\omega_{\p_2}}{\sqrt{\Lambda_2}}
		\text{,}
	\end{align*}
	which implies 
	\begin{align*} 
		(|\p_1+\p_2|^2+|\p_1|^2+|\p_2|^2)\left[\omega_{\p_1}^m+\omega_{\p_2}^m\right]
		&\leq 3 (|\p_1|^2+|\p_2|^2)\left[\omega_{\p_1}^m+\omega_{\p_2}^m\right] \\
		&\lesssim (\omega_{\p_1}^2+\omega_{\p_2}^2)\left[\omega_{\p_1}^m+\omega_{\p_2}^m\right]
		\lesssim \left[\omega_{\p_1}^{m+2}+\omega_{\p_2}^{m+2}\right]
		\text{.}
	\end{align*}

	Therefore
	\begin{equation}\label{Propo:C12:E3}
		\mathcal{A}_1 \lesssim  \iint_{\mathbb{R}^{2d}} g_1g_2  
		\Big [\cE^{m+2}_{\p_1}+\cE^{m+2}_{\p_2} \Big] d\p_1d\p_2
		\lesssim \mathcal{M}_{m+2}[g]
		\text{.}
	\end{equation}

		\medskip
		\noindent\textit{Step 3. Estimate of $\mathcal{A}_2$.}
		
Using the resonant condition  $\p_2 = \p- \p_1$,  we obtain 
$$\cE_{\p_2} =\sqrt{\Lambda_1+\Lambda_2 |\p_2|^2} \leq \sqrt{\Lambda_1+\Lambda_2 (|\p_1|+|\p|)^2}$$
$$\le 2\sqrt{\Lambda_1+\Lambda_2 |\p|^2} + 2\sqrt{\Lambda_1+\Lambda_2 |\p_1|^2} = 2\cE_{\p} +2\cE_{\p_1},$$ 
which implies
$$\cE_{\p_2}^m  \lesssim \cE_{\p}^m + \cE_{\p_1}^m.$$ 

Thus, we obtain 
$$
\begin{aligned}
	\mathcal{A}_2 \lesssim & \iiint_{\mathbb{R}^{3d}} \delta(\p-\p_1-\p_2)
	(|\p|^2+|\p_1|^2+|\p_2|^2)
	gg_1  \Big [\cE^m_{\p}+\cE^m_{\p_1} \Big] d\p d\p_1d\p_2.
\end{aligned}
$$
Integrating by $\p_2$ and using the definition of the Dirac function $\delta(\p-\p_1-\p_2)$
$$
\begin{aligned}
	\mathcal{A}_2 \lesssim & \ \iint_{\mathbb{R}^{2d}} 
	(|\p-\p_1|^2+|\p|^2+|\p_1|^2)
	gg_1  \Big [\cE^m_{\p}+\cE^m_{\p_1} \Big] d\p d\p_1.
\end{aligned}
$$
This yields the following bound on $\mathcal{A}_2$ 
\begin{equation}\label{Propo:C12:E4}
	\begin{aligned}
		\mathcal{A}_2\lesssim \ \iint_{\mathbb{R}^{2d}} 
		gg_1 \Big [\cE^{m+2}_{\p}+\cE^{m+2}_{\p_1} \Big] d\p d\p_1
		\lesssim  \mathcal{M}_{m+2}[g]
		\text{.}
	\end{aligned}
\end{equation}
Combining \eqref{Propo:C12:E1}--\eqref{Propo:C12:E4}, 
we get \eqref{Propo:C12:1} so the conclusion of the Lemma \ref{Propo:C12} follows. 
	\end{proof}

\section{ $L^1_m$ $(m\geq 0)$ estimates}

\begin{proposition}\label{Propo:MomentsPropa} 
	Let $m \ge 0$ and $\gamma>2$. For any nonnegative initial data $f_0(\p)$ satisfying 
	\[
	\int_{\mathbb{R}^d} f_0(\p) \cE_\p^m \, d\p < \infty,
	\]
	there is a constant $\widetilde{C}=\widetilde{C}(\Lambda_1, \Lambda_2,\gamma,\nu)>1$ 
	depending only on 
	$\Lambda_1, \Lambda_2,\gamma,\nu$ and independent of $m$, such that
	\begin{equation}\label{EE-bound}
		\mathcal{M}_m[f](t) \leq e^{\widetilde{C}(\Lambda_1,\Lambda_2,\gamma,\nu)t} 
		\int_{\mathbb{R}^d} f_0(\p) \cE_\p^m \, d\p.
	\end{equation}
\end{proposition}

\begin{proof}[Proof of Proposition \ref{Propo:MomentsPropa}]
	Using $\varphi = \cE_\p^m$ as a test function in \eqref{WeakTurbulenceInitial}, we have
	\[
	\frac{d}{dt} \mathcal{M}_m[f] + 2\nu \mathcal{M}_m[|\p|^\gamma f] 
	= \frac{d}{dt} \int_{\mathbb{R}^d} f(t,\p) \cE_\p^m \, d\p + 2\nu \int_{\mathbb{R}^d} 
	|\p|^\gamma f(t,\p) \cE_\p^m \, d\p 
	= \int_{\mathbb{R}^d} \mathbb{C}[f](t,\p) \cE_\p^m \, d\p.
	\]
	
	Applying Lemma \ref{Propo:C12}, we obtain
	\begin{equation}\label{E}
		\frac{d}{dt} \mathcal{M}_m[f] 
		+2\nu\int_{\mathbb{R}^d} |\p|^\gamma f(t,\p)\omega_\p^m\, d\p 
		=\int_{\mathbb{R}^d} \mathbb{C}[f](t,\p)\omega_\p^n\, d\p
		\lesssim \mathcal{M}_{m+2} [f]
		\text{,}
	\end{equation}
	which implies 
	\begin{align*} 
		\frac{d}{dt} \mathcal{M}_m[f] 
		\leq\int_{\mathbb{R}^d}f(t,\p)\ \omega_\p^m(C\omega_\p^2-2\nu|\p|^\gamma)\,d\p
		\text{.} 
	\end{align*}
	Observe that as $\gamma>2$, 
	\begin{align*} 
		C\omega_\p^2-2\nu |\p|^\gamma
		=C(\Lambda_1+\Lambda_2|\p|^2)-2\nu|\p|^\gamma 
	\end{align*}
	is bounded above by a constant $\widehat{C}(\Lambda_1,\Lambda_2,\gamma)$
	depending on $\Lambda_1$, $\Lambda_2$, and $\gamma$. 
	Therefore, 
	\begin{align*} 
		\frac{d}{dt}\mathcal{M}_{m}(t) 
		\leq \widetilde{C}(\Lambda_1,\Lambda_2,\gamma,\nu)
		\int_{\mathbb{R}^d} f(t,\p)\omega_\p^m\,d\p 
	\end{align*}
	for $\widetilde{C}=2\widehat{C}$.
	Inequality \eqref{EE-bound} then follows from Gr\"onwall's inequality.
\end{proof}

\section{Bounds of the solution}
\begin{proposition}\label{Propo:MassLowerBound} 
	Let $f_0$ be positive initial data in $L^1(\mathbb{R}^d)$, $\gamma>2$, 
	and let $f\in L^1(\mathbb{R}^d)$ be the corresponding positive strong solution of \eqref{WeakTurbulenceInitial}. 
	Then, we have
	\begin{equation}\label{Propo:MassLowerBound:1} 
		\mathbb{C}[f]=\mathbb{C}_{\mathrm{gain}}[f]-\mathbb{C}_{\mathrm{loss}}[f]
		\geq -\mathbb{C}_{\mathrm{loss}}[f]
		\geq -\left(A_1|\p|^2+A_2\right)e^{\widetilde{C}t}f
		\text{,}
	\end{equation}
	where $\widetilde{C}$ is the constant from Proposition \ref{Propo:MomentsPropa}
	and $A_1$, $A_2$ are positive constants that depend on 
	$\|f_0\|_{L^1_2}$, $\Lambda_1$, $\Lambda_2$. 
\end{proposition}

\begin{proof}
	Since
	\begin{equation*}
		\begin{aligned}
			\mathbb{C}[f] &= \iint_{\mathbb{R}^{d}\times \mathbb{R}^{d}} |V_{\p,\p_1,\p_2}|^2 \delta(\p-\p_1-\p_2) \mathcal{L}_f(\omega_\p-\omega_{\p_1}-\omega_{\p_2}) (f_1 f_2 - 2 f f_1) \, d\p_1 d\p_2 \\
			&\quad + 2 \iint_{\mathbb{R}^{d}\times \mathbb{R}^{d}} |V_{\p_1,\p,\p_2}|^2 \delta(\p_1-\p-\p_2) \mathcal{L}_f(\omega_{\p_1}-\omega_{\p}-\omega_{\p_2}) (-f f_2 + f f_1 + f_1 f_2) \, d\p_1 d\p_2,
		\end{aligned}
	\end{equation*}
	we split $\mathbb{C}[f]$ as
	\[
	\mathbb{C}[f] = \mathbb{C}_{\mathrm{gain}}[f] - \mathbb{C}_{\mathrm{loss}}[f],
	\]
	where
	\begin{equation}\label{Propo:MassLowerBound:E1}
		\begin{aligned}
			-\mathbb{C}_{\mathrm{loss}}[f] &= -2 f \int_{\mathbb{R}^{d}\times \mathbb{R}^{d}} |V_{\p,\p_1,\p_2}|^2 \delta(\p-\p_1-\p_2) \mathcal{L}_f(\omega_\p-\omega_{\p_1}-\omega_{\p_2}) f_1 \, d\p_1 d\p_2 \\
			&\quad -2 f \int_{\mathbb{R}^{d}\times \mathbb{R}^{d}} |V_{\p_1,\p,\p_2}|^2 \delta(\p_1-\p-\p_2) \mathcal{L}_f(\omega_{\p_1}-\omega_{\p}-\omega_{\p_2}) f_2 \, d\p_1 d\p_2 \\
			&=: -\mathcal{B}_1 - \mathcal{B}_2.
		\end{aligned}
	\end{equation}
	
	We now discard the gain term and estimate the loss term from below.
	
	\medskip
	\noindent
	{\it Estimating $\mathcal{B}_1$:}
	
	Using the Dirac delta to reduce the integral, we have
	\[
	\mathcal{B}_1 = 2 f \int_{\mathbb{R}^{d}} |V_{\p,\p_1,\p-\p_1}|^2 \mathcal{L}_f(\omega_\p-\omega_{\p_1}-\omega_{\p-\p_1}) f_1 \, d\p_1.
	\]
	The kernel satisfies
	\begin{align*}
	|V_{\p,\p_1,\p-\p_1}|^2 \mathcal{L}_f(\omega_\p-\omega_{\p_1}-\omega_{\p-\p_1})
	&\leq \frac{\mathfrak{C}}{\gamma_\p+\gamma_{\p_1}+\gamma_{\p_2}}
	(|\p|+|\p_1|+|\p-\p_1|)^2 \\
	\leq \frac{\mathfrak{C}}{3\mathfrak{c}_1\mathfrak{c}_2}
	(|\p|+|\p_1|+|\p-\p_1|)^2 
	&\leq \frac{\mathfrak{C}}{\mathfrak{c}_1\mathfrak{c}_2}
	(|\p|^2+|\p_1|^2)
	\text{.}
	\end{align*}
	Thus,
	\begin{equation}\label{Propo:MassLowerBound:E2}
		\mathcal{B}_1\leq 
		\frac{2\mathfrak{C}}{\mathfrak{c}_1\mathfrak{c}_2}f\left(
			|\p|^2\int_{\mathbb{R}^d} f_1\,d\p_1 
			+\int_{\mathbb{R}^d} |\p_1|^2f_1\,d\p_1
		\right)
		=\frac{2\mathfrak{C}}{\mathfrak{c}_1\mathfrak{c}_2}
		\left(|\p|^2\mathcal{M}_0[f]+\mathcal{M}_2[f]\right)f
		\text{.}
	\end{equation}
	
	\medskip
	\noindent
	{\it Estimating $\mathcal{B}_2$:}
	
	Similarly, using $\delta(\p_1-\p-\p_2)$ we obtain
	\[
	\mathcal{B}_2 = f \int_{\mathbb{R}^{d}} |V_{\p+\p_2, \p, \p_2}|^2 \mathcal{L}_f(\omega_{\p+\p_2}-\omega_\p-\omega_{\p_2}) f_2 \, d\p_2.
	\]
	The kernel can be bounded as
	\[
	|V_{\p+\p_2, \p, \p_2}|^2 \mathcal{L}_f(\omega_{\p+\p_2}-\omega_\p-\omega_{\p_2}) \le
	\frac{\mathfrak{C}}{\mathfrak{c}_1\mathfrak{c}_2}
	(|\p|^2+|\p_2|^2)
	\text{.}
	\]
	Hence,
	\begin{equation}\label{Propo:MassLowerBound:E3}
		\mathcal{B}_2\leq 
		\frac{2\mathfrak{C}}{\mathfrak{c}_1\mathfrak{c}_2}f\left(
			|\p|^2\int_{\mathbb{R}^d} f_1\,d\p_1 
			+\int_{\mathbb{R}^d} |\p_1|^2f_1\,d\p_1
		\right)
		=\frac{2\mathfrak{C}}{\mathfrak{c}_1\mathfrak{c}_2}
		\left(|\p|^2\mathcal{M}_0[f]+\mathcal{M}_2[f]\right)f
		\text{.}
	\end{equation}
	
	\medskip
	\noindent
	Combining \eqref{Propo:MassLowerBound:E1}--\eqref{Propo:MassLowerBound:E3}
	and applying Proposition \ref{Propo:MomentsPropa}, we obtain
	\begin{equation}\label{Propo:MassLowerBound:E4}
		-\mathbb{C}_{\mathrm{loss}}[f] \ge
		-(A_1|\p|^2+A_2)e^{\widetilde{C}t}f
		\text{,}
	\end{equation}
	where $\widetilde{C}=\widetilde{C}(\Lambda_1,\Lambda_2,\gamma,\nu)$ is computed in 
	Proposition \ref{Propo:MomentsPropa}
	and $A_1$, $A_2$ depend on $\|f_0\|_{L^1_2}$, $\Lambda_1$, and $\Lambda_2$
	by Proposition \ref{Propo:MomentsPropa}. 
	This proves \eqref{Propo:MassLowerBound:1} and the proof is complete.
\end{proof}

\section{Estimates for $\mathbb{C}[f]$}\label{Sec:HolderEstimate}

\begin{proposition}\label{Propo:HolderC12} 
	Let $M, m \ge 0$, 
	and suppose that $S_{M}$ is a bounded subset of $L^1_{m+2}(\RR^d)$ satisfying, 
	for all $g \in S_{M}$, 
	\[
	\|g\|_{m+2} \le M
	\quad\text{and }
	g \geq 0
	\text{.}
	\] 
	Then, for all $g,h \in S_{M}$,
	\begin{equation}\label{Propo:HolderC12:1} 
		\|\mathbb{C}[g] - \mathbb{C}[h]\|_{L^1_m} 
		\lesssim \|g - h\|_{m+2}^{\frac12}, 
	\end{equation}
	where the constants depend only on $M$ and $m$.
\end{proposition}

We first establish the following lemma.

\begin{lemma}\label{Lemma:Holder} 
	Let $M, m \ge 0$, and suppose that $S_{M}$ is 
	as in Proposition \ref{Propo:HolderC12}.
	Then, for all $g,h \in S_{M}$,
	\begin{equation}\label{Q-L1bound} 
		\|\mathbb{C}[g] - \mathbb{C}[h]\|_{L^1_m} 
		\lesssim \|g - h\|_{m+2},
	\end{equation}
	where the constants depend only on $M$ and $m$.
\end{lemma}

\begin{proof} 
We first compute the difference between $\mathbb{C}[g]$ and $\mathbb{C}[h]$:
\[
\mathbb{C}[g] - \mathbb{C}[h] = 
\iint_{\mathbb{R}^{2d}} \Big[ \mathcal{N}_{\p,\p_1,\p_2}[g] - \mathcal{N}_{\p,\p_1,\p_2}[h] 
- 2 \big( \mathcal{N}_{\p_1,\p,\p_2}[g] - \mathcal{N}_{\p_1,\p,\p_2}[h] \big) \Big] d\p_1 d\p_2,
\]
and its $L^1_m$-norm:
\[
\begin{aligned}
	\|\mathbb{C}[g]-\mathbb{C}[h]\|_{L^1_m} &= \int_{\mathbb{R}^d} \cE_\p^m \, |\mathbb{C}[g](\p) - \mathbb{C}[h](\p)| \, d\p \\
	&\le \iiint_{\mathbb{R}^{3d}} \cE_\p^m \, |\mathcal{N}_{\p,\p_1,\p_2}[g] - \mathcal{N}_{\p,\p_1,\p_2}[h]| \, d\p \, d\p_1 \, d\p_2 \\
	&\quad + 2 \iiint_{\mathbb{R}^{3d}} \cE_\p^m \, |\mathcal{N}_{\p_1,\p,\p_2}[g] - \mathcal{N}_{\p_1,\p,\p_2}[h]| \, d\p \, d\p_1 \, d\p_2 \\
	&= \iiint_{\mathbb{R}^{3d}} |\mathcal{N}_{\p,\p_1,\p_2}[g] - \mathcal{N}_{\p,\p_1,\p_2}[h]| \, \big( \cE_\p^m + \cE_{\p_1}^m + \cE_{\p_2}^m \big) \, d\p \, d\p_1 \, d\p_2.
\end{aligned}
\]

Therefore, we obtain the following estimate:
\begin{equation}\label{Lemma:Holder:E1} 
	\|\mathbb{C}[g]-\mathbb{C}[h]\|_{L^1_m} \le \mathcal{D}_1 + \mathcal{D}_2,
\end{equation}
where
\begin{equation}\label{Lemma:Holder:E2} 
	\begin{aligned}
		\mathcal{D}_1 &:= \iiint_{\mathbb{R}^{3d}} |V_{\p,\p_1,\p_2}|^2 \, \delta(\p-\p_1-\p_2) \, \Big| \mathcal{L}_g(\omega_\p - \omega_{\p_1} - \omega_{\p_2}) g_1 g_2 \\
		&\quad - \mathcal{L}_h(\omega_\p - \omega_{\p_1} - \omega_{\p_2}) h_1 h_2 \Big| \, \big( \cE_\p^m + \cE_{\p_1}^m + \cE_{\p_2}^m \big) \, d\p \, d\p_1 \, d\p_2, \\
		\mathcal{D}_2 &:= 2 \iiint_{\mathbb{R}^{3d}} |V_{\p_1,\p,\p_2}|^2 \, \delta(\p_1 - \p - \p_2) \, \Big| \mathcal{L}_g(\omega_{\p_1} - \omega_{\p} - \omega_{\p_2}) g g_2 \\
		&\quad - \mathcal{L}_h(\omega_{\p_1} - \omega_{\p} - \omega_{\p_2}) h h_2 \Big| \, \big( \cE_\p^m + \cE_{\p_1}^m + \cE_{\p_2}^m \big) \, d\p \, d\p_1 \, d\p_2.
	\end{aligned}
\end{equation}

{\it Estimating $\mathcal{D}_1$.}

Set the quantity  inside the triple integral of $\mathcal{D}_1$ after dropping $\Big( \cE_\p^m + \cE_{\p_1}^m + \cE_{\p_2}^m\Big)$ to be $\mathbb{D}_1$
\begin{equation*}
	\begin{aligned}
		\mathbb{D}_1 \ := &\  |V_{\p,\p_1,\p_2}|^2\delta(\p-\p_1-\p_2) \Big|\mathcal{L}_g(\omega_\p -\omega_{\p_1}-\omega_{\p_2})g_1g_2-\mathcal{L}_h(\omega_\p -\omega_{\p_1}-\omega_{\p_2})h_1h_2\Big|,
	\end{aligned}
\end{equation*}
which can be bounded as, using the triangle inequality,
\begin{equation*}
	\begin{aligned}
		\mathbb{D}_1 \ \le &\  |V_{\p,\p_1,\p_2}|^2\delta(\p-\p_1-\p_2) \mathcal{L}_g(\omega_\p -\omega_{\p_1}-\omega_{\p_2})|g_1g_2-h_1h_2|\\
		\  &\ + |V_{\p,\p_1,\p_2}|^2\delta(\p-\p_1-\p_2)\Big| \mathcal{L}_g(\omega_\p -\omega_{\p_1}-\omega_{\p_2})-\mathcal{L}_h(\omega_\p -\omega_{\p_1}-\omega_{\p_2})\Big||h_1h_2|\\
		=: &\ \mathbb{D}_{11}+\mathbb{D}_{12}.
	\end{aligned}
\end{equation*}
Let us now study $\mathbb{D}_{11}$ in details. Using   the triangle inequality
$$|g_1g_2-h_1h_2|\le |g_1||g_2-h_2|+|h_2||g_1-h_1|,$$
yields 
$$
\begin{aligned}
	\mathbb{D}_{11}\ \lesssim 
	& \ \delta(\p-\p_1-\p_2)
	(|\p|+|\p_1|+|\p_2|)^2
	\mathcal{L}_g(\omega_\p -\omega_{\p_1}-\omega_{\p_2})|g_1||g_2-h_2|\\
	\ &\ + \delta(\p-\p_1-\p_2)
	(|\p|+|\p_1|+|\p_2|)^2
	\mathcal{L}_g(\omega_\p -\omega_{\p_1}-\omega_{\p_2})|h_2||g_1-h_1|\\
	\lesssim 
	& \ \delta(\p-\p_1-\p_2) 
	(|\p|+|\p_1|+|\p_2|)^2 \Big[
	|g_1| |g_2 - h_2 |
	+ |h_2 | |g_1 - h_1 | \Big]
	\text{.}
\end{aligned}
$$
Here, the estimate 
\begin{align*} 
	\mathcal{L}_g (\omega_\p - \omega_{\p_1}-\omega_{\p_2})
	\leq \frac{1}{\gamma_\p + \gamma_{\p_1} + \gamma_{\p_2}}
	\leq \frac{1}{3\mathfrak{c}_1 \mathfrak{c}_2}
\end{align*}
was used.

Multiplying the above inequality with $\Big( \cE_\p^m + \cE_{\p_1}^m + \cE_{\p_2}^m\Big)$ and integrating in $\p$, $\p_1$ and $\p_2$, we obtain 
$$
\begin{aligned}
	&\ \iiint_{\mathbb{R}^{3d}}\mathbb{D}_{11}\Big( \cE_\p^m + \cE_{\p_1}^m + \cE_{\p_2}^m\Big)d\p d\p_1d\p_2\\
	\ \lesssim &  \ \iiint_{\mathbb{R}^{3d}}{\delta(\p-\p_1-\p_2)}
	(|\p|+|\p_1|+|\p_2|)^2
	\left[|g_1||g_2-h_2| + |h_2||g_1-h_1|\right]\\
	&\times\Big( \cE_\p^m + \cE_{\p_1}^m + \cE_{\p_2}^m\Big)d\p d\p_1d\p_2.
\end{aligned}
$$
By the resonant condition  $\p = \p_1 + \p_2$, 
$$
\begin{aligned}
	&\ \iiint_{\mathbb{R}^3}\mathbb{D}_{11}\Big( \cE_\p^m + \cE_{\p_1}^m + \cE_{\p_2}^m\Big)d\p d\p_1d\p_2\\
	\ \lesssim 
	& \ \iint_{\mathbb{R}^{2d}}
	(|\p_1+\p_2|+|\p_1|+|\p_2|)^2
	\left[|g_1||g_2-h_2| + |h_2||g_1-h_1|\right]
	\Big(\cE_{\p_1}^m + \cE_{\p_2}^m\Big)d\p_1d\p_2,
\end{aligned}
$$
where the inequality $\cE_{\p_1+\p_2}^m\lesssim \cE_{\p_1}^m + \cE_{\p_2}^m\text{,}$
proved in Proposition \ref{Propo:C12}, was used to bound $\cE_{\p}^m + \cE_{\p_1}^m + \cE_{\p_2}^m$ by 
$C\left(\cE_{\p_1}^m + \cE_{\p_2}^m\right)$.
Since 
\begin{align*} 
	(|\p_1+\p_2|+|\p_1|+|\p_2|)^2 (\omega_{\p_1}^m+\omega_{\p_2}^m)
	&\lesssim (|\p_1|^2+|\p_2|^2)(\omega_{\p_1}^m+\omega_{\p_2}^m)\\
	&\lesssim (\omega_{\p_1}^2+\omega_{\p_2}^2)(\omega_{\p_1}^m+\omega_{\p_2}^m)
	\lesssim \omega_{\p_1}^{m+2}+\omega_{\p_2}^{m+2}
\end{align*}
as in the proof of \eqref{Propo:C12:E3}, we find
$$
\begin{aligned}
	&\ \iiint_{\mathbb{R}^{3d}}\mathbb{D}_{11}\Big( \cE_\p^m + \cE_{\p_1}^m + \cE_{\p_2}^m\Big)d\p d\p_1d\p_2\\
	\ \lesssim & \ \iint_{\mathbb{R}^{2d}}\left[|g_1||g_2-h_2|
	+ |h_2||g_1-h_1|\right]\Big(\cE_{\p_1}^{m+2}+\cE_{\p_2}^{m+2}\Big)d\p_1d\p_2,
\end{aligned}
$$
which immediately yields
\begin{equation}\label{Lemma:Holder:E3} 
	\begin{aligned}
		&\ \iiint_{\mathbb{R}^{3d}}\mathbb{D}_{11}\Big( \cE_\p^m + \cE_{\p_1}^m + \cE_{\p_2}^m\Big)d\p d\p_1d\p_2\\
		\ \lesssim & \|g-h\|_{L^1_{m+2}}\left(\|g\|_{L^1}+\|g\|_{L^1_{m+2}}+\|h\|_{L^1}+\|h\|_{L^1_{m+2}}\right)\\
		\ \lesssim & \|g-h\|_{L^1_{m+2}}\left(\|g\|_{L^1_{m+2}}+\|h\|_{L^1_{m+2}}\right).
	\end{aligned}
\end{equation} 
Now, let us look at $\mathbb{D}_{12}$, which can be written for 
$\Delta=\omega_\p-\omega_{\p_1}-\omega_{\p_2}$ as
$$
\begin{aligned}
	\mathbb{D}_{12} \ = & \ \delta(\p-\p_1-\p_2)
	(|\p|+|\p_1|+|\p_2|)^2
	|h_1h_2| \times\left|
		\mathcal{L}_g(\Delta)-\mathcal{L}_h(\Delta)
	\right|\\
	= & \ \delta(\p-\p_1-\p_2)
	(|\p|+|\p_1|+|\p_2|)^2
	|h_1h_2|\\ &\times\left|
		\frac{\Gamma^g_{\p,\p_1,\p_2}}{\Delta^2+(\Gamma^g_{\p,\p_2,\p_2})^2} -
		\frac{\Gamma^h_{\p,\p_1,\p_2}}{\Delta^2+(\Gamma^h_{\p,\p_2,\p_2})^2}
	\right| \\
	= & \ \delta(\p-\p_1-\p_2)
	(|\p|+|\p_1|+|\p_2|)^2
	|h_1h_2|\\ &\times\left|
		\frac{
			(\Gamma^g_{\p,\p_1,\p_2}-\Gamma^h_{\p,\p_1,\p_2})(\Delta^2-\Gamma^g_{\p,\p_1,\p_2}\Gamma^h_{\p,\p_1,\p_2})}{
			(\Delta^2+(\Gamma^g_{\p,\p_1,\p_2})^2)(\Delta^2+(\Gamma^h_{\p,\p_1,\p_2})^2)
		}
	\right|
	\text{.}
\end{aligned}
$$

It follows from the Cauchy-Schwarz inequality that
$$
(\Delta^2 + (\Gamma^g_{\p,\p_1,\p_2})^2)(\Delta^2+(\Gamma^h_{\p,\p_1,\p_2})^2)
\geq (\Delta^2+\Gamma^g_{\p,\p_1,\p_2}\Gamma^h_{\p,\p_1,\p_2})
|\Delta^2-\Gamma^g_{\p,\p_1,\p_2}\Gamma^h_{\p,\p_1,\p_2}|
\text{,}
$$
from which we obtain the following estimate on $\mathbb{D}_{12}$
$$
\begin{aligned}
	\mathbb{D}_{12} 
	&\leq \delta(\p-\p_1-\p_2)|h_1 h_2|
	(|\p|+|\p_1|+|\p_2|)^2
	\frac{|\Gamma^g_{\p,\p_1,\p_2}-\Gamma^h_{\p,\p_1,\p_2}|}{
		\Delta^2+\Gamma^g_{\p,\p_1,\p_2}\Gamma^h_{\p,\p_1,\p_2}
	}\\
	&\leq \delta(\p-\p_1-\p_2)|h_1 h_2|
	(|\p|+|\p_1|+|\p_2|)^2
	\frac{|\Gamma^g_{\p,\p_1,\p_2}-\Gamma^h_{\p,\p_1,\p_2}|}{
		\Gamma^g_{\p,\p_1,\p_2}\Gamma^h_{\p,\p_1,\p_2}
	}
	\text{.}
\end{aligned}
$$
As $x\mapsto \max\{x,\mathfrak{c}_2\}$ is 1-Lipschitz, 
the numerator can be bounded as 
\begin{equation*}
\begin{aligned}
	|\Gamma^g_{\p,\p_1,\p_2}-\Gamma^h_{\p,\p_1,\p_2}|
	\leq&\mathfrak{c}_1\Big(
		|\max\{|\p|g,\mathfrak{c}_2\}-\max\{|\p|h,\mathfrak{c}_2\}|\\
		&+|\max\{|\p_1|g_1,\mathfrak{c}_2\}-\max\{|\p_1|h_1,\mathfrak{c}_2\}| \\
		&+|\max\{|\p_2|g_2,\mathfrak{c}_2\}-\max\{|\p_2|h_2,\mathfrak{c}_2\}|
	\Big) \\ 
	&\leq \mathfrak{c}_1 \left(
		|\p| |g-h| + |\p_1||g_1-h_1|+|\p_2||g_2-h_2|
	\right)
	\text{,}
\end{aligned}
\end{equation*}
yielding an upper bound for $\mathbb{D}_{12}$:
\begin{equation}\label{Lemma:Hoelder:ED12}
\begin{aligned}
	\mathbb{D}_{12}\lesssim & \
	\delta(\p-\p_1-\p_2)|h_1 h_2|
	(|\p|+|\p_1|+|\p_2|)^2
	\frac{
		|\p||g-h|+|\p_1||g_1-h_1|+|\p_2||g_2-h_2|
	}{\Gamma^g_{\p,\p_1,\p_2}\Gamma^h_{\p,\p_1,\p_2}}\\
	=&\delta(\p-\p_1-\p_2)|h_1 h_2|
	(|\p|+|\p_1|+|\p_2|)^2
	\frac{|\p||g-h|}{\Gamma^g_{\p,\p_1,\p_2}\Gamma^h_{\p,\p_1,\p_2}}\\
	&+\delta(\p-\p_1-\p_2)|h_1 h_2|
	(|\p|+|\p_1|+|\p_2|)^2
	\frac{|\p_1||g_1-h_1|}{\Gamma^g_{\p,\p_1,\p_2}\Gamma^h_{\p,\p_1,\p_2}}\\
	&+\delta(\p-\p_1-\p_2)|h_1 h_2|
	(|\p|+|\p_1|+|\p_2|)^2
	\frac{|\p_2||g_2-h_2|}{\Gamma^g_{\p,\p_1,\p_2}\Gamma^h_{\p,\p_1,\p_2}}\\
	=:& \mathbb{D}_{120}+\mathbb{D}_{121}+\mathbb{D}_{122}
	\text{.}
\end{aligned}
\end{equation}
Now, we split $\mathbb{D}_{12}$ using \eqref{Lemma:Hoelder:ED12}
and estimate the integrals of these terms separately.
Starting with $\mathbb{D}_{121}$,
We integrate by $\p$ and use the resonant condition $\p=\p_1+\p_2$ to obtain 
\begin{equation*}
\begin{aligned}
	&\iiint_{\mathbb{R}^{3d}} \mathbb{D}_{121}(\omega_\p^m+\omega_{\p_1}^m+\omega_{\p_2}^m)
	\,d\p d\p_1 d\p_2 \\
	&= \iint_{\mathbb{R}^{2d}} (|\p_1+\p_2|+|\p_1|+|\p_2|)^2 |h_1 h_2| 
	\frac{|\p_1||g_1-h_1|}{\Gamma^h_{\p,\p_1,\p_2}\Gamma^g_{\p,\p_1,\p_2}}
	(\omega_{\p_1+\p_2}^m+\omega_{\p_1}^m+\omega_{\p_2}^m)
	\, d\p_1 d\p_2 
	\text{.}
\end{aligned}
\end{equation*}
Using the estimates 
\begin{equation}\label{Lemma:Hoelder:E3-0}
\begin{aligned}
	\Gamma^g_{\p,\p_1,\p_2}\geq \gamma_\p \geq \mathfrak{c}_1\mathfrak{c}_2 
	\text{,}\quad 
	\Gamma^h_{\p,\p_1,\p_2}\geq \gamma_{\p_1}\geq |h_1||\p_1|
\end{aligned}
\end{equation}
yields 
\begin{equation}\label{Lemma:Hoelder:E3-1}
\begin{aligned}
	&\iiint_{\mathbb{R}^{3d}} \mathbb{D}_{121}(\omega_\p^m+\omega_{\p_1}^m+\omega_{\p_2}^m)
	\,d\p d\p_1 d\p_2 \\
	&\lesssim \iint_{\mathbb{R}^{2d}} (|\p_1+\p_2|+|\p_1|+|\p_2|)^2 |h_2||g_1-h_1| 
	(\omega_\p^m+\omega_{\p_1}^m+\omega_{\p_2}^m) \, d\p_1 d\p_2 
	\text{.}
\end{aligned}
\end{equation}
Similarly, we also obtain 
\begin{equation}\label{Lemma:Hoelder:E3-2}
\begin{aligned}
	&\iiint_{\mathbb{R}^{3d}} \mathbb{D}_{122}(\omega_\p^m+\omega_{\p_1}^m+\omega_{\p_2}^m)
	\,d\p d\p_1 d\p_2 \\
	&\lesssim \iint_{\mathbb{R}^{2d}} (|\p_1+\p_2|+|\p_1|+|\p_2|)^2 |h_1||g_2-h_2| 
	(\omega_\p^m+\omega_{\p_1}^m+\omega_{\p_2}^m) \, d\p_1 d\p_2 
	\text{.}
\end{aligned}
\end{equation}
Combining \eqref{Lemma:Hoelder:E3-1}, \eqref{Lemma:Hoelder:E3-2} and applying the same procedure
used to derive \eqref{Lemma:Holder:E3} leads to 
\begin{equation}\label{Lemma:Hoelder:E3-3}
\begin{aligned}
	\iiint_{\mathbb{R}^{3d}} & \ (\mathbb{D}_{121}+\mathbb{D}_{122})
	(\omega_\p^m+\omega_{\p_1}^m+\omega_{\p_2}^m)\, d\p d\p_2 d\p_2  \\
	& \lesssim \|g-h\|_{L_{m+2}^1} \left( \|g\|_{L_{m+2}^1}+\|h\|_{L_{m+2}^1}
	\right)
\end{aligned}
\end{equation}

Now we estimate the integral containing $\mathbb{D}_{120}$.
Using the resonant condition $\p=\p_1+\p_2$, 
\begin{equation}\label{Lemma:Hoelder:E3-3-0}
\begin{aligned}
	\iiint_{\mathbb{R}^{3d}} \ & \mathbb{D}_{120}(\omega_\p^m+\omega_{\p_1}^m+\omega_{\p_2}^m)
	\,d\p d\p_1 d\p_2 \\
	=& \iiint_{\mathbb{R}^{3d}} (|\p|+|\p_1|+|\p_2|)^2 
	\delta(\p-\p_1-\p_2)|h_1 h_2| 
	\frac{|\p_1+\p_2| |g-h|}{\Gamma^g_{\p,\p_1,\p_2}\Gamma^h_{\p,\p_1,\p_2}} \\
	&\times (\omega_\p^m+\omega_{\p_1}^m+\omega_{\p_2}^m) \,d\p d\p_1 d\p_2 \\
	\leq & \ \iiint_{\mathbb{R}^{3d}} (|\p|+|\p_1|+|\p_2|)^2 
	\delta(\p-\p_1-\p_2)|h_1 h_2| 
	\frac{|\p_1| |g-h|}{\Gamma^g_{\p,\p_1,\p_2}\Gamma^h_{\p,\p_1,\p_2}} \\
	&\times (\omega_\p^m+\omega_{\p_1}^m+\omega_{\p_2}^m) \,d\p d\p_1 d\p_2 \\
	+ & \ \iiint_{\mathbb{R}^{3d}} (|\p|+|\p_1|+|\p_2|)^2 
	\delta(\p-\p_1-\p_2)|h_1 h_2| 
	\frac{|\p_2| |g-h|}{\Gamma^g_{\p,\p_1,\p_2}\Gamma^h_{\p,\p_1,\p_2}} \\
	&\times (\omega_\p^m+\omega_{\p_1}^m+\omega_{\p_2}^m) \,d\p d\p_1 d\p_2 \\
	=:& \mathcal{D}_{11}' + \mathcal{D}_{12}'
	\text{.}
\end{aligned}
\end{equation}
To estimate $\mathcal{D}_{11}'$, 
we integrate in $\p_1$ and use \eqref{Lemma:Hoelder:E3-0} to get,
following the proof of \eqref{Lemma:Hoelder:E3-3},
\begin{equation}\label{Lemma:Hoelder:E3-3-1}
\begin{aligned}
	\mathcal{D}_{11}'
	&\lesssim \iint_{\mathbb{R}^{2d}} (|\p|^2+|\p_2|^2) |h_2| |g-h|
	(\omega_\p^m + \omega_{\p_2}^m)\, d\p d\p_2 \\ 
	& \lesssim \|g-h\|_{L_{m+2}^1} \left( \|g\|_{L_{m+2}^1}+\|h\|_{L_{m+2}^1}
	\right)
	\text{.}
\end{aligned}
\end{equation}
Similarly, we obtain 
\begin{equation}\label{Lemma:Hoelder:E3-3-2}
\begin{aligned}
	\mathcal{D}_{12}'
	&\lesssim \iint_{\mathbb{R}^{2d}} (|\p|^2+|\p_1|^2) |h_2| |g-h|
	(\omega_\p^m + \omega_{\p_1}^m)\, d\p d\p_1 \\ 
	& \lesssim \|g-h\|_{L_{m+2}^1} \left( \|g\|_{L_{m+2}^1}+\|h\|_{L_{m+2}^1}
	\right)
	\text{.}
\end{aligned}
\end{equation}

Combining \eqref{Lemma:Holder:E3}, \eqref{Lemma:Hoelder:ED12}
and \eqref{Lemma:Hoelder:E3-3}--\eqref{Lemma:Hoelder:E3-3-2} yields
\begin{equation}\label{Lemma:Holder:E5}
\begin{aligned}
	\mathcal{D}_1\  \lesssim &\ \|g-h\|_{L^1_{m+2}},
\end{aligned}
\end{equation}
where the constant in the above inequality depends on 
$\left(\|g\|_{L^1_{m+1}}+\|h\|_{L^1_{m+1}}\right)$.

{\it Estimating $\mathcal{D}_2$.}

The proof of estimating $\mathcal{D}_2$ follows exactly the same argument used in the previous estimate. We omit some details and give only the main estimates in the sequel. First, define the quantity  inside the triple integral of $\mathcal{D}_2$ after dropping $\Big( \cE_\p^m + \cE_{\p_1}^m + \cE_{\p_2}^m\Big)$ to be $\mathbb{D}_2$.
\begin{equation*}
	\begin{aligned}
		\mathbb{D}_2 \ := &\  |V_{\p_1,\p,\p_2}|^2\delta(\p_1-\p-\p_2) \Big|\mathcal{L}_g(\omega_{\p_1} -\omega_{\p}-\omega_{\p_2})gg_2-\mathcal{L}_h(\omega_{\p_1} -\omega_{\p}-\omega_{\p_2})hh_2\Big|,
	\end{aligned}
\end{equation*}
which, by the triangle inequality, can be bounded as
\begin{equation*}
	\begin{aligned}
		\mathbb{D}_2 \ \lesssim &\  |V_{\p_1,\p,\p_2}|^2\delta(\p_1-\p-\p_2) \mathcal{L}_g(\omega_{\p_1} -\omega_{\p}-\omega_{\p_2})|gg_2-hh_2|\\
		\  &\ + |V_{\p_1,\p,\p_2}|^2\delta(\p_1-\p-\p_2)\Big| \mathcal{L}_g(\omega_{\p_1} -\omega_{\p}-\omega_{\p_2})-\mathcal{L}_h(\omega_{\p_1} -\omega_{\p}-\omega_{\p_2})\Big||hh_2|.
	\end{aligned}
\end{equation*}
Define the two terms on the right hand side of the above inequality to be $\mathbb{D}_{21}$ and $\mathbb{D}_{22}$, respectively.

The same argument used in Step 1 can be employed, implying the following estimate: 
$$ 
\begin{aligned}
	\mathbb{D}_{21} \ \lesssim & \ {\delta(\p-\p_1-\p_2)}
	\left(|\p|+|\p_1|+|\p_2|\right)^2
	\left( |g||g_2-h_2| + |h_2||g-h| \right)
	\text{.}
\end{aligned}
$$
Multiplying the above by $\Big( \cE_\p^m + \cE_{\p_1}^m + \cE_{\p_2}^m\Big)$ and integrating in 
$\p$, $\p_1$ and $\p_2$ yields
\begin{equation}\label{Lemma:Holder:E6} 
	\begin{aligned}
		&\ \iiint_{\mathbb{R}^{3d}}\mathbb{D}_{21}\Big( \cE_\p^m + \cE_{\p_1}^m + \cE_{\p_2}^m\Big)d\p d\p_1d\p_2\\
		\ \lesssim & \ \left(\|g-h\|_{L^1}+\|g-h\|_{L^1_{m+2}}\right),
	\end{aligned}
\end{equation} 
where the constant depends on $\left(\|g\|_{L^1_{m+2}}+\|h\|_{L^1_{m+2}}\right)$.

Now, similar to $\mathbb{D}_{12}$, $\mathbb{D}_{22}$  can be bounded as
$$
\begin{aligned}
	\mathbb{D}_{22} 
	\ \lesssim & \ |hh_2|\delta(\p_1-\p-\p_2)\frac{|\Gamma^g_{\p,\p_1,\p_2}-\Gamma^h_{\p,\p_1,\p_2}|}{\Gamma^g_{\p,\p_1,\p_2}\Gamma^h_{\p,\p_1,\p_2}}.
\end{aligned}
$$
The same argument used in \eqref{Lemma:Holder:E3} can be applied and we then obtain
\begin{equation}\label{Lemma:Holder:E7}
	\begin{aligned}
		&\ \iiint_{\mathbb{R}^{3d}}\mathbb{D}_{22} \Big( \cE_\p^m+ \cE_{\p_1}^m + \cE_{\p_2}^m\Big)d\p d\p_1d\p_2
		\lesssim \left(\|g-h\|_{L^1}+\|g-h\|_{L^1_{m+2}}\right),
	\end{aligned}
\end{equation}
where the constant depends on $\left(\|g\|_{L^1_{m+2}}+\|h\|_{L^1_{m+2}}\right)$.
\\
Combining \eqref{Lemma:Holder:E6} and \eqref{Lemma:Holder:E7} yields
\begin{equation}\label{Lemma:Holder:E8}
	\begin{aligned}
		\mathcal{D}_2\  \lesssim \|g-h\|_{L^1}+\|g-h\|_{L^1_{m+2}}\lesssim  \|g-h\|_{L^1_{m+2}}.
	\end{aligned}
\end{equation}
Putting the two estimates \eqref{Lemma:Holder:E5} and \eqref{Lemma:Holder:E8} together with \eqref{Lemma:Holder:E1} and \eqref{Lemma:Holder:E2}, the conclusion of the Lemma then follows. 
\end{proof}

\begin{proof}[Proof of Proposition \ref{Propo:HolderC12}] 
The proposition now follows straightforwardly from the previous lemma. 
Indeed, by the boundedness of $g,h$ in $L^1_1 \cap L^1_{m+2}$, we obtain 
$$
\begin{aligned}
	\| g-h\|_{L^1_{m+2}}  &\le \| g-h\|_{L^1_{m+2}}^{\frac12} 
	\left(\|g\|_{L^1_{m+2}}+\|h\|_{L^1_{m+2}}\right)^{\frac12}
	\lesssim \| g-h\|_{L^1_{m+2}}^{\frac12} 
	\text{.}
\end{aligned}$$
Therefore, we have 
$$\|\mathbb{C}[g]-\mathbb{C}[h]\|_{L^1_m}  \lesssim\| g-h\|_{L^1_{m+2}}^{\frac12} $$
which holds for all $m\geq 0 $. The proposition follows. 
\end{proof}

\section{Proof of Theorem \ref{Theorem:Main}}\label{Sec:Main}
We shall apply Theorem \ref{Theorem:ODE} to \eqref{WeakTurbulenceInitial}, which can be written as
\[
\partial_t f = \mathbb{Q}[f], \qquad \mathbb{Q}[f] := \mathbb{C}[f] - 2\nu |\p|^\gamma f.
\]
Fix $m>1$, and define the Banach spaces $\mathfrak{E} = L^1_m(\mathbb{R}^d)$ and $\mathfrak{F} = L^1_{m+3}(\mathbb{R}^d)$, endowed with the norms
\[
\|f\|_{\mathfrak{E}} := \|f\|_{L^1_m}, \qquad \|f\|_* := \|f\|_{L^1_{m+3}}.
\]
We also define
\[
|f|_* := \mathcal{M}_{m+3}[f].
\]
Then we have
\[
|f|_* \le \|f\|_*, \quad \forall f \in \mathfrak{F}, \qquad |f+g|_* \le |f|_* + |g|_*, \quad \forall f,g \in \mathfrak{F},
\]
\[
\Lambda |f|_* = |\Lambda f|_*, \quad \forall f \in \mathfrak{F}, \Lambda \in \mathbb{R}_+,
\]
and
\[
|f|_* = \|f\|_{L^1_{m+3}}, \quad \forall f \in \Omega_T.
\]
Moreover, condition \eqref{LesbegueDominated} is automatically satisfied due to the Lebesgue dominated convergence theorem and Theorem 1.2.7 in \cite{BadialeSerra:SEE:2011}.

Clearly, $\Omega_T$ is a bounded and closed set with respect to the norm $\|\cdot\|_*$. By Proposition \ref{Propo:MomentsPropa}, for $f_0 \in \Omega_0 \subset \Omega_T$, solutions to \eqref{WeakTurbulenceInitial} remain in $\Omega_T$. Thus, it suffices to verify the three conditions $(\mathscr{A})$, $(\mathscr{B})$, and $(\mathscr{C})$ of Theorem \ref{Theorem:ODE}. Then, Theorem \ref{Theorem:Main} follows as a direct consequence of Theorem \ref{Theorem:ODE}. 

Notice that the continuity condition $(\mathscr{A})$ follows directly from Proposition \ref{Propo:HolderC12}, so it remains to verify $(\mathscr{B})$ and $(\mathscr{C})$.

\subsection{Condition $(\mathscr{B})$: Subtangent condition}\label{Subtangent}

Let $f$ be an arbitrary element of the set $\Omega_T$. It suffices to prove the following claim: for all $\epsilon>0$, there exists $h_*>0$, depending on $f$ and $\epsilon$, such that 
\begin{equation}\label{claim}
	B(f + h \mathbb{Q}[f], h\epsilon) \cap \Omega_T \neq \emptyset, \qquad 0 < h < h_*.
\end{equation}

For $R>0$, let $\chi_R(\p)$ denote the characteristic function of the ball $B(0,R)$, and define  
\begin{equation}\label{def-wR} 
	w_R := f + h \mathbb{Q}[f_R], \qquad f_R(\p) := \chi_R(\p) f(\p)
	\text{.}
\end{equation}
We shall show that for each $R>0$, there exists $h_R>0$
to be determined later such that $w_R \in \Omega_T$ 
for all $0 < h \le h_R$. 

Clearly, $w_R \in L^1(\RR^d) \cap L^1_{m+3}(\RR^d)$,
and we now verify the conditions (S1) and (S2) in \eqref{Cons}.  

~\\
{\bf Condition (S1): Positivity of the set $\Omega_T$.} Note that we can write 
\[
\mathbb{C}[f] = \mathbb{C}_\mathrm{gain}[f] - \mathbb{C}_\mathrm{loss}[f],
\]
with $\mathbb{C}_\mathrm{gain}[f] \ge 0$ and $\mathbb{C}_\mathrm{loss}[f] = f \vartheta[f]$. 
Since $f_R$ is compactly supported, 
it follows from Proposition \ref{Propo:MassLowerBound} that 
$\chi_R \vartheta[f_R]$ is bounded by a universal constant 
$(A_1 R^2+A_2)e^{\widetilde{C}T}$ . 
Here, $A_1$ and $A_2$ depend on $\Lambda_1$, $\Lambda_2$ and $\|f_0\|_{L^1_2}$,
where the norm of $f_0$ is bounded using $\varsigma$ as $f_0\in B_*(0,\varsigma)$.
Hence,
\[
\begin{aligned}
	w_R &= f + h \Big( \mathbb{C}[f_R] - 2 \nu |\p|^\gamma f_R \Big) \\
	&\ge f - h f_R \Big((A_1R^2+A_2)e^{\widetilde{C}T} + 2\nu R^\gamma\Big),
\end{aligned}
\]
which is nonnegative for sufficiently small $h$, specifically
\[
h < \frac{h_R}{2} := \frac{1}{2\left((A_1R^2+A_2)e^{\widetilde{C}T} + 2\nu R^\gamma\right)}.
\]

Let us check \eqref{Coercivity} for $\eta < R$. By Lemma \ref{Propo:C12},
\[
\begin{aligned}
|w_R - f|_* &= h |\mathbb{C}[f_R] - 2\nu |\p|^\gamma f_R|_*  
=h\left(\int_{\mathbb{R}^d} \mathbb{C}[f_R] \ \omega_\p^{m+3}\,d\p 
-2\nu\int_{\mathbb{R}^d} |\p|^\gamma f_R \ \omega_\p^{m+3}\,d\p
\right) \\
&\leq h\left(C \int_{\mathbb{R}^d} f_R \ \omega_\p^{m+5}\,d\p 
-2\nu\int_{\mathbb{R}^d} |\p|^\gamma \omega_\p^{m+3}\,d\p
\right) 
= h\left(\int_{\mathbb{R}^d} f_R \ \omega_\p^{m+3}
\left( C\omega_\p^2-2\nu |\p|^\gamma \right)\,d\p
\right) 
\text{,}
\end{aligned}
\]
where $C\omega_\p^2-2\nu|\p|^\gamma$ is bounded above as in the proof of 
Proposition \ref{Propo:MomentsPropa}. 
This yields 
\begin{equation}\label{K1R0}
	\left|\frac{w_R - f}{h}\right|_* 
	\le \frac{\widetilde{C}}{2}\int_{\mathbb{R}^d} f_R \ \omega_\p^{m+3}\,d\p 
	\leq \frac{\theta_*}{2}\|f\|_* 
	\text{,}
\end{equation}
where $\widetilde{C}$ is the constant from Proposition \ref{Propo:MomentsPropa}.

{\bf Condition (S2): Upper bound of the set $\Omega_T$.} 
By Proposition \ref{Propo:MomentsPropa}, $\|f\|_* < (2\varsigma + 1) e^{\theta_* T}$. 
Since 
\[
\lim_{h \to 0} \|f - w_R\|_* = 0,
\]
we can choose $h_*$ small enough so that for $0 < h < h_*$,
\[
\|w_R\|_* < (2\varsigma + 1) e^{\theta_* T}.
\]

This proves the claim \eqref{claim} and hence verifies condition $(\mathscr{B})$.

\subsection{Condition $(\mathscr{C})$: One-side Lipschitz condition}\label{Lipschitz}

By the Lebesgue dominated convergence theorem, we have
\[
\begin{aligned}
	\big[\varphi, \phi \big] 
	&= \lim_{h \rightarrow 0^-} h^{-1} \big( \| \phi + h \varphi \|_E - \| \phi \|_E \big) \\
	&= \lim_{h \rightarrow 0^-} h^{-1} \int_{\RR^d} \big( |\phi + h \varphi| - |\phi| \big) (\cE_\p + \cE_\p^m) \, d\p \\
	&\le \int_{\RR^d} \varphi(\p) \, \mathrm{sign}(\phi(\p)) (\cE_\p + \cE_\p^m) \, d\p.
\end{aligned}
\]

Recalling that $\mathbb{Q}[f] = \mathbb{C}[f] - 2\nu |\p|^\gamma f$, we estimate
\[
\begin{aligned}
	\big[ \mathbb{Q}[f] - \mathbb{Q}[g], f - g \big] 
	&\le \int_{\RR^d} \big[ \mathbb{Q}[f](\p) - \mathbb{Q}[g](\p) \big] \, \mathrm{sign}((f-g)(\p)) \, \cE_\p^m \, d\p \\
	&\le\|\mathbb{C}[f]-\mathbb{C}[g]\|_{\mathfrak{E}}-2\nu\|\, |\p|^\gamma (f-g)\,\|_{\mathfrak{E}}.
\end{aligned}
\]

Using Lemma \ref{Lemma:Holder} and recalling that $\|\cdot\|_{\mathfrak{E}} = \|\cdot\|_{L^1_m}$, we obtain
\[
\|\mathbb{C}[f] - \mathbb{C}[g]\|_{\mathfrak{E}} \le C_m \| f - g \|_{L^1_m}.
\]

Since $C|\p|^m - 2 \nu |\p|^{m+\gamma}$ is always bounded by $C'|\p|^m$ for some $C'>0$, it follows that
\[
\big[ \mathbb{Q}[f] - \mathbb{Q}[g], f - g \big] \le C_m \| f - g \|_{\mathfrak{E}}.
\]

Thus, condition $(\mathscr{C})$ is satisfied. This completes the proof of Theorem \ref{Theorem:Main}.

\section{Proof of Theorem \ref{Theorem:ODE}}\label{Appendix}
The proof is divided into four parts.

\textbf{Part 1:}  
By our assumptions, $\Omega_T$ is bounded by a constant $C_S$ in the norm $\|\cdot\|$, and due to the H\"older continuity of $\mathcal{Q}[u]$, we have
\[
\|\mathcal{Q}[u]\|\le C_\mathcal{Q}, \quad \forall u \in \Omega_T.
\]  
For an element $u \in \Omega_0 \subset \Omega_T$, there exists $\xi_u>0$ such that for $0<\xi<\xi_u$,
\[
B(u+\xi \mathcal{Q}[u],\delta)\cap \Omega_T \setminus \{u+\xi \mathcal{Q}[u]\} \neq \emptyset
\]
for $\delta$ sufficiently small.  

For a fixed $u$ and $\epsilon\in (0,1)$, there exists $\xi>0$ such that if $\|u-v\|\le (C_\mathcal{Q}+1)\xi$, then $\|\mathcal{Q}(u)-\mathcal{Q}(v)\|\le \epsilon/2$. Let $z \in B\big(u+\xi \mathcal{Q}[u], \frac{\epsilon \xi}{2}\big) \cap \Omega_T \setminus \{u+\xi \mathcal{Q}[u]\}$ satisfy
\[
\left|\frac{z-u}{\xi}\right|_* \le \frac{\theta_*}{2}\|u\|_*, 
\]
and define
\[
t \mapsto \Theta(t) = u + \frac{t(z-u)}{\xi}, \quad t \in [0, \xi].
\]

We also have the following upper bound on $\Theta$:
\[
\begin{aligned}
	\|\Theta(t)\|_* &= |\Theta(t)|_* = \left| u + \frac{t(z-u)}{\xi} \right|_* \\
	&\le |u|_* + \left| \frac{t(z-u)}{\xi} \right|_* \le |u|_* + |u|_* \frac{t\theta_*}{2} \\
	&= \|\Theta(0)\|_* \left(1 + \frac{t\theta_*}{2}\right),
\end{aligned}
\]
which implies
\begin{equation}\label{Theorem:ODE:E4}
	\|\Theta(t)\|_* \le (\|\Theta(0)\|_* + 1) e^{\theta_* t} - 1 < (2\varsigma + 1) e^{\theta_* t}.
\end{equation}

Thus, $\Theta$ maps $[0, \xi]$ into $\Omega_T$. It is straightforward to see that
\[
\|\Theta(t) - u\| \le \left\|\frac{t(z-u)}{\xi}\right\| \le \xi \|\mathcal{Q}[u]\| + \frac{\epsilon \xi}{2} < (C_\mathcal{Q}+1)\xi,
\]
which implies
\[
\|\mathcal{Q}[\Theta(t)] - \mathcal{Q}[u]\| \le \frac{\epsilon}{2}, \quad \forall t \in [0,\xi].
\]  
Combining this with
\[
\|\dot{\Theta}(t) - \mathcal{Q}[u]\| = \left\|\frac{z-u}{\xi} - \mathcal{Q}[u]\right\| \le \frac{\epsilon}{2},
\]
we obtain
\begin{equation}\label{Theorem:ODE:E1}
	\|\dot{\Theta}(t) - \mathcal{Q}[\Theta(t)]\| \le \epsilon, \quad \forall t \in [0, \xi].
\end{equation}

\textbf{Part 2:}  
Let $\Theta$ be a solution to \eqref{Theorem:ODE:E1} on $[0, \xi]$ constructed in Part 1. Using the same procedure, we can extend $\Theta$ to the interval $[\xi, \xi + \xi']$.  

The same arguments that led to \eqref{Theorem:ODE:E4} imply
\[
\|\Theta(\xi+t)\|_* \le (\|\Theta(\xi)\|_* + 1) e^{\theta_* t} - 1
\text{,}\quad t \in [0, \xi'].
\]  
Combining this with \eqref{Theorem:ODE:E4}, we have
\begin{equation}\label{Theorem:ODE:E5}
	\begin{aligned}
		\|\Theta(\xi+t)\|_* &\le \big( (\|\Theta(0)\|_* + 1) e^{\theta_* \xi} - 1 + 1 \big) e^{\theta_* t} - 1 \\
		&= (\|\Theta(0)\|_* + 1) e^{\theta_*(\xi+t)} - 1 \\
		&< (2\varsigma + 1) e^{\theta_*(\xi+t)},
	\end{aligned}
\end{equation}
where the last inequality follows from $\varsigma \ge 1$.

\textbf{Part 3:}  
From Part 1, there exists a solution $\Theta$ to \eqref{Theorem:ODE:E1} on an interval $[0,\xi]$. We proceed as follows:

\begin{itemize}
	\item \textit{Step 1:} Suppose we have constructed a solution $\Theta$ of \eqref{Theorem:ODE:E1} on $[0,\tau]$ with $\tau<T$, where $\Theta(0) \in \Omega_0 \cap B_*(O, \varsigma)$. By Part 2, $\Theta(\tau) \in \Omega_\tau$. 
	Using the same procedure as in Part 1 
	and applying \eqref{Theorem:ODE:E4} 
	and \eqref{Theorem:ODE:E5}, 
	the solution $\Theta$ can be extended to $[\tau, \tau+h_\tau]$ with $\tau+h_\tau \le T$.  
	
	\item \textit{Step 2:} Suppose we have constructed $\Theta$ on a sequence of intervals $[0, \tau_1]$, $[\tau_1, \tau_2]$, $\dots$, $[\tau_n, \tau_{n+1}]$, $\dots$. Since the increasing sequence $\{\tau_n\}$ is bounded by $T$, it converges to a limit, denoted $\tau$. Moreover, we have
	\begin{equation}\label{Theorem:ODE:E5b}
		\begin{aligned}
			\|\Theta(t)\|_* &\le (\|\Theta(0)\|_* + 1) e^{\theta_* t} - 1 
			< (2\varsigma+1)e^{\theta_* t}, & \forall t \in [0, \tau)
			\text{.}
		\end{aligned}
	\end{equation}
	
	Since $\|\mathcal{Q}(\Theta)\|$ is bounded by $C_\mathcal{Q}$ on each interval $[\tau_n, \tau_{n+1}]$, it follows that $\|\dot{\Theta}\|$ is bounded by $\epsilon + C_\mathcal{Q}$ on $[0, \tau)$. Therefore, $\Theta(\tau)$ can be defined as the limit of $\Theta(\tau_n)$ in the norm $\|\cdot\|$.  
	Together with \eqref{LesbegueDominated} and the fact that $\Omega_\tau$ is closed in $\|\cdot\|_*$, this implies that $\Theta$ is a solution of \eqref{Theorem:ODE:E1} on $[0, \tau]$, and \eqref{Theorem:ODE:E5b} holds on $[0, \tau]$ as well.
\end{itemize}

Consequently, if a solution $\Theta$ is defined on $[0, T_0)$ with $T_0 < T$, it can be extended to $[0, T_0]$. If $[0, T_0]$ is the maximal interval where $\Theta$ is defined (by Steps 1 and 2), then $\Theta$ can be further extended to $[T_0, T_0+T_h]$. This implies $T_0 = T$, so $\Theta$ is defined on the entire interval $[0, T]$.

\textbf{Part 4:}  
Finally, consider a sequence of solutions $\{u^\epsilon\}$ to \eqref{Theorem:ODE:E1} on $[0, T]$. We show that this sequence is Cauchy.  

Let $\{u^\epsilon\}$ and $\{v^\epsilon\}$ be two such sequences. Since $u^\epsilon$ and $v^\epsilon$ are affine on $[0, T]$, and by the one-sided Lipschitz condition, we have for a.e. $t \in [0, T]$,
\begin{equation*}
	\begin{aligned}
		\frac{d}{dt} \|u^\epsilon(t)-v^\epsilon(t)\|
		&= \big[\dot{u}^\epsilon(t)-\dot{v}^\epsilon(t), {u}^\epsilon(t) - {v}^\epsilon(t)\big] \\
		&\le \big[\mathcal{Q}[u^\epsilon(t)] - \mathcal{Q}[v^\epsilon(t)], u^\epsilon(t)-v^\epsilon(t)\big] + 2\epsilon \\
		&\le L \|u^\epsilon(t) - v^\epsilon(t)\| + 2\epsilon,
	\end{aligned}
\end{equation*}

which implies
\[
\|u^\epsilon(t) - v^\epsilon(t)\| \le \frac{2\epsilon}{L} e^{L T}.
\]

Letting $\epsilon \to 0$, we obtain $u^\epsilon \to u$ uniformly on $[0, T]$. It follows immediately that $u$ is a solution to \eqref{Theorem_ODE_Eq}.

\section{Conclusion}

We formulated a three--wave kinetic equation for stratified fluids in the ocean, incorporating a physically motivated resonance--broadening operator and collision kernel, and proceeded to prove global existence and uniqueness of strong solutions in \(L^1_m(\mathbb{R}^d)\). 

We considered nonlinear interactions between three wavenumbers, where two wavenumbers are much larger in magnitude than the third wavenumber. 
When oceanographers study the kinetic equation in this limit, they typically make the scale separation large enough so that the kinetic equation can be represented by a diffusion equation, in the induced diffusion limit \cite{mccomasbretherton77}.
In the current work, our approach is different.  We take into account the near resonant interactions without taking the extreme scale separation.  Thus, we anticipate that the current formulation of the kinetic equation may be more accurate than the traditional diffusion approximation of the kinetic equation.

The present work advances the near--resonant program by replacing an acoustic--oriented broadening approximation \cite{GambaSmithBinh} with one better suited to oceanographic settings \cite{polzin2017oceanic}, where Garrett--Munk--type phenomenology emerges \cite{garrett1975space,garrett1979internal}.

\bibliographystyle{plain}
\bibliography{QuantumBoltzmann}

\end{document}